\documentclass[12pt]{article}
\usepackage[latin9]{inputenc}
\usepackage{amsmath}
\usepackage{amssymb}
\usepackage{graphicx}
\usepackage{esint}
\usepackage[authoryear]{natbib}
\usepackage[unicode=true, bookmarks=false, breaklinks=false,pdfborder={0 0 1},backref=section,colorlinks=false] {hyperref}
\hypersetup{colorlinks,citecolor=darkblue,filecolor=black,linkcolor=darkblue,urlcolor=blue,pdfstartview=FitH}

\makeatletter
\usepackage{euscript}
\usepackage{amsfonts}
\usepackage{caption}
\usepackage{soul}
\usepackage{color}
\usepackage{fullpage}

\definecolor{webgreen}{rgb}{0,0.4,0}
\definecolor{webbrown}{rgb}{0.6,0,0}
\definecolor{purple}{rgb}{0.5,0,0.25}
\definecolor{darkblue}{rgb}{0,0,0.7}
\definecolor{darkred}{rgb}{0.7,0,0}
\definecolor{darkgreen}{rgb}{0,0.7,0}

\newcommand{\ignore}[1]{}

\newtheorem{prop}{{\bf Proposition}}

\newtheorem{theorem}{{\bf Theorem}}
\newtheorem{defn}{{\bf Definition}}
\newtheorem{obs}{{\bf Observation}}
\newtheorem{claim}{{\bf Claim}}

\newtheorem{fact}{{\bf Fact}}

\newtheorem{remark}{{\bf Remark}}

\usepackage[normalem]{ulem}

\newenvironment{proof}{\noindent {\bf \sl Proof\/}:\enspace}{\hfill $\blacksquare{}$ \vspace{12pt}}

\newcommand{\lord}{\succ_{{\rm LD}}}
\title{{\bf Teacher transfers: \\ equalizing deficits across schools}\thanks{For comments and suggestions, we are grateful to Tommy Andersson, Anna Bogomolnaia, Julien Combe, Bhaskar Dutta, Albin Erlanson, Manshu Khanna, Vijay Krishna, Herve Moulin, Abhiroop Mukhopadhyay, Sanket Patil, Antonio Penta, Debraj Ray, Olivier Tercieux, Utku Unver and Mike Woodford, and seminar participants at conferences and workshops where the paper was presented.}}
\author{Debasis Mishra~~~~~~Soumendu Sarkar~~~~~~Arunava Sen \\
Jay Sethuraman~~~~~~~Sonal Yadav\thanks{Mishra: Indian Statistical Institute, Delhi, \texttt{dmishra@isid.ac.in}. Sarkar: Delhi School of Economics, \texttt{soumendu@econdse.org}. Sen: Indian Statistical Institute, \texttt{asen@isid.ac.in}. Sethuraman: Columbia University, \texttt{jay@ieor.columbia.edu}. Yadav: University of Liverpool, \texttt{sonal.crk@gmail.com}}}

\begin{document}
\maketitle 
\begin{abstract}
The Right to Free and Compulsory Education Act (2009) (RTE) of the Government of India prescribes student-teacher ratios for state-run schools.
One method advocated by the Act to achieve its goals is the redeployment of teachers from surplus to deficit (in teacher strength) schools. We consider a model where teachers can either remain in their initially assigned schools or be transferred to a deficit school in their acceptable set. The planner's objective is specified in terms of the post-transfer deficit vector that can be achieved. We show that there exists a transfer whose post-transfer deficit vector Lorenz dominates all achievable post-transfer deficit vectors. We provide a two-stage algorithm to derive the Lorenz-dominant post-transfer deficit vector, and show that this algorithm is strategy-proof for teachers.
\end{abstract}



\newpage{}

\section{Introduction}
\label{sec:introduction}
The equitable allocation of resources is a fundamental concern in economics. It is well-known that various equity criteria can be satisfied in models where resources are perfectly divisible. However, several problems are  {\sl discrete} in nature and only allow for integral solutions. Our objective in this paper is to show that an existing equitable solution for perfectly divisible resources can be extended to a discrete problem of practical interest.

The Right to Education Act in India, 2009 (henceforth, RTE) was enacted by the Union Government to improve the quality of school education
in India.\footnote{For an overview of the act, see \href{https://dsel.education.gov.in/rte}{\texttt{https://dsel.education.gov.in/rte}}.} The Act identified student-teacher ratios in schools as a key determinant of the quality of education and clearly mandated the minimum number of teachers to be 
deployed relative to student strength. For historical reasons, some
schools had teachers in excess of the desired student-teacher ratio while other schools had a deficit of teachers. In addition to fresh hiring, the Act mandated the 
transfer of teachers from surplus to deficit schools. According to Section $22$ of the RTE: 
\begin{quote}
{\sl  ``...The sanctioned strength of teachers in
a school shall be notified by the Central Government, appropriate Government or the
local authority...The Central Government, appropriate Government or the local authority,
as the case may be, shall .... redeploy teachers of schools having strength in excess of the sanctioned strength....''}
\end{quote}

The implementation of the RTE Act was largely left to the various states in the Union of India. The government of the 
state of Haryana was especially proactive in the matter of teacher transfers and its policies are enunciated in 
two documents (\citet{TTP16} and \citet{TTP23}). According to the vision statements 
of both documents, the aim of its policies are 
\begin{quote}
{\sl `` ... to ensure equitable, demand based distribution of teachers/heads to protect academic interest of students and optimise job satisfaction amongst its employees in a fair and transparent manner''}.
\end{quote}

Our paper is an attempt to design a theoretical framework to achieve the goals articulated in the vision statement.\footnote{\citet{sharan2022teachers} and \citet{agarwal2018redistributing} contain empirical analyses of teacher transfers in Haryana: see also Section \ref{sec:othrel}.} We formulate the teacher transfer problem
as a network flow problem where teachers have trichotomous preferences. We show that fairness and efficiency of outcomes can be defined in a natural and unambiguous way. Finally, we show that these ``fair outcomes'' can be achieved by a strategy-proof mechanism. Though we consider teacher transfer as our application, we note that our model and results apply more broadly to environments involving transfer of resources. For instance, civil servants assigned to states are transferred periodically; (multi-tasking) workers are transferred across departments in firms; doctors are transferred within states to meet shortages in hospitals.

We consider a set of teachers and a set of schools. Each school is either a  surplus school or a deficit school, depending on the number of teachers currently 
assigned to it and its target number of teachers. Prior to transfers, every surplus school is associated  with a strictly positive integer that is its surplus. Similarly, every deficit school is associated with a strictly positive integer that is its deficit. 
A transfer is a reassignment of teachers from surplus to deficit schools satisfying several requirements. Therefore, only the teachers assigned to surplus schools can be transferred; each such teacher has a set of acceptable deficit schools to which she can be transferred. Each teacher must either be transferred to an acceptable deficit school\footnote{In Section \ref{sec:ext}, we consider the case where a teacher from a surplus school can be transferred to an acceptable surplus school. We do not allow the transfer of teachers from deficit schools in our model.}
or remain in the surplus school to which she belongs. This is an individual rationality requirement for teachers. No surplus school can become a deficit school
nor can a deficit school become a surplus school post-transfer. The former  is an individual rationality requirement for surplus schools and the latter, a natural fairness assumption. 
Finally, we do not permit fractional transfers of teachers.  This is a natural assumption in our context
since teachers are always assigned to a single school in practice. 

Every transfer leads to a post-transfer deficit vector of deficit schools. We assume (in broad consonance with the Haryana policy documents)  that the planner's
objectives are expressed solely in terms of the post-transfer deficit vector. One can conceive of the planner having one of several objectives.
A utilitarian planner may wish to minimize 
the aggregate sum of deficits or transfer as many 
teachers as possible. An egalitarian planner 
may wish to ``minimize the worst deficits''.
A welfare maximizing planner may wish to minimize total loss for a continuous and convex loss function of the deficits. These solutions 
will not coincide in an arbitrary  problem.  Our main result, Theorem \ref{theo:main} shows that there exists an achievable post-transfer deficit vector that Lorenz dominates all other achievable post-transfer deficits. Such a transfer simultaneously reconciles the objectives of utilitarian, egalitarian and 
welfare maximizing planners. A novel feature of our result is that we demonstrate the 
existence of a Lorenz dominant solution under integrality constraints. 

We provide a two-stage algorithm to find a Lorenz dominant vector. We formulate our problem as one of finding an integer flow in a network. In the first stage of the algorithm, we consider a relaxed version of the network flow problem where flows are allowed to be non-negative real numbers. We then use an existing algorithm by \cite{megiddo1974optimal} and \cite{dutta1989concept} (henceforth, MDR), to find Lorenz dominant flows in this network. These flows are not integral.\footnote{Consider a problem in which there is one surplus school with three teachers and two deficit schools. Suppose every teacher finds both the deficit schools acceptable and, the current deficits of both the deficit schools are equal, say 3. Then, the MDR algorithm will generate flows which will transfer 1.5 teachers to each of the deficit schools, leading to the post-transfer deficit vector (1.5,1.5).} The second stage of the algorithm uses the MDR solution to find an appropriate rounding of these flows which corresponds to the Lorenz dominant deficit vector of the teacher transfer problem. This requires solving a maximum flow problem of an augmented network constructed using the original network and the solution of the MDR algorithm. We note that a Lorenz dominant transfer can be computed efficiently (see Remark \ref{rem:eff} in Section \ref{sec:rounding}).

From a methodological viewpoint, our work extends the work of \citet{dutta1989concept}. They consider convex cooperative games and their (the MDR) algorithm finds a Lorenz dominant core allocation --- core is non-empty in a convex game with elegant extreme point properties~\citep{S71}. Consider a convex cooperative game whose characteristic function is integral, which is the case for the teacher transfer problem. Our work addresses the following question: does there exist an {\sl integral} core allocation that Lorenz dominates every other integral core allocation. We answer this question in the affirmative for the specific class of convex games 
induced by the network flow problems in \citet{megiddo1974optimal}. We show that the output of the MDR algorithm can be used to do a careful rounding to get an integral vector. As long as this rounding remains in the core, it is an integral Lorenz dominant core allocation. 

We note that transfers could be made more efficient if teachers could be transferred from a surplus school to another surplus school while maintaining the restrictions that surplus schools do not become deficit schools and deficit schools do not become surplus schools. 
Consider the case in which there are three schools: two surplus schools $s$ and $s'$, each with a surplus of one; and a deficit school $d$ with a deficit of two. Suppose there is exactly one teacher $t$ currently assigned to $s$ who is willing to be transferred, but this teacher finds only $s'$ acceptable; and suppose teachers $t'_1$ and $t'_2$ are currently assigned to $s'$ but are both willing to be transferred to $d$. In our basic model, the deficit at $d$ can only be reduced to one because exactly one of the teachers assigned to school $s'$ can be transferred to $d$. However, if transfers between surplus schools is permitted, the deficit at $d$ can be reduced to zero: transfer $t$ from $s$ to $s'$; and transfer both $t'_1$ and $t'_2$ to $d$. In Section \ref{sec:ext}, which discusses extensions of our basic model, we show how our results can be extended to cover this case.

We consider another extension where teachers are not homogeneous. They can be qualified to teach different subjects
and  a teacher may be able to teach more than one subject. We show that such a model can be accomodated 
in our framework by a suitable adaptation of the underlying network. In Section \ref{sec:ext}, we show that our main result holds 
in this setting as well. 

We also consider the teacher transfer problem where the acceptable set of every teacher is private information.
The current assignment of teachers and the surplus and deficit vectors are, however, publicly observable.  We assume
that a teacher's preferences are trichotomous. Her top indifference class consists of the deficit schools she finds 
acceptable; the second indifference class is the surplus school that she currently 
belongs to; and the third class is the set of non-acceptable deficit schools. A transfer mechanism must therefore
rely on the reports of acceptable sets by the teachers. We show that the two-stage algorithm, supplemented by a consistent tie-breaking condition, is strategy-proof for teachers.

Our work is most closely related to the theoretical literature on matching with constraints and redistribution of matched resources. The key distinguishing feature of our work is that we seek an integral post transfer deficit vector that is Lorenz dominant. \citet{BM04} and \citet{bochet2012balancing} have considered Lorenz dominance as an objective in specific matching models. We outline the relationship between these papers and ours next.

\citet{BM04} consider two-sided matching problems in which both men and women have dichotomous preferences. For every agent (both a man and a woman), a {\sl random} matching generates the probability of an acceptable match. The paper constructs a random matching where this probability vector is Lorenz dominant. Like us, their proposed matching is the output of the \citet{dutta1989concept} algorithm for a convex cooperative game. There are however, several significant differences between their paper and ours. The first is that they allow for random matching while we are  interested in (integral) deterministic transfers. As a result, we have to supplement the MDR algorithm with a rounding stage to achieve a Lorenz dominant transfer. The second is that our problem is not a two-sided matching problem. In our model, teachers are employed in surplus schools which have constraints on the number of teachers they can release for transfer. As we shall show in Example $1$,  this fact plays an important role in our analysis. Consequently, a teacher is not an independent agent since her matching is tied to the surplus in the school where she is currently assigned. This is a departure from \cite{BM04}.
Finally, our objective is to equalize deficits, different from that of equalizing transfers.

\citet{bochet2012balancing} analyse a model in which a divisible, non-disposable 
homogeneous commodity has to be reallocated between agents with single-peaked preferences. Agents are either suppliers 
or demanders, and transfers between a supplier and demander are feasible only if they are linked. 
These links form a bipartite graph. The teachers in our model can be regarded as the commodity being transferred 
from surplus to deficit schools. However, teachers have preferences over deficit schools and cannot be treated as a homogenous commodity. They cannot also be treated as a divisible commodity since fractional appointments are not permitted in our model. Finally as mentioned earlier, the transfer 
of teachers is constrained by the surplus of the schools in their initial assignment. 

The paper is organized as follows. Section \ref{sec:model} presents the model and the main result. 
Section \ref{sec:network} describes the procedure to find a Lorenz dominant transfer. Section \ref{sec:SP}
discusses the strategy-proofness issue, while Section \ref{sec:ext} considers extensions of the basic model. Section \ref{sec:othrel} reviews additional related
literature. Finally, Appendix \ref{app:mdr} provides a proof of Lorenz dominance of the output of the MDR algorithm.

\section{The model}
\label{sec:model}

We consider a set of schools partitioned into two subsets: the set of deficit schools $D$ where there are fewer teachers than the 
number mandated by the RTE act and the set of surplus schools $S$ with more teachers than the mandated number. 
Each surplus school has an initial assignment of teachers. Let $T$ denote the set of all teachers assigned to surplus schools.
Only the teachers in $T$ are candidates for being transferred; teachers employed in non-surplus schools cannot be transferred and so play no further role in 
our basic model.\footnote{The RTE explicitly requires teachers to be transferred from surplus to deficit schools. See 
the quotation from the Act in the Introduction. The Haryana government also followed this practice - see Points (ix) and (x) on Page $7$ in \citet{TTP16}.}

We denote generic deficit schools, surplus schools, and teachers by $d_k$, $s_j$, and $t_i$ respectively. 
Each teacher $t_i$ belongs to a unique surplus 
school $O(t_i)\in S$. This is the initial assignment 
of $t_i$. 

Each teacher $t_i$ has a set of {\sl acceptable} deficit schools $A(t_i) \subseteq D$. As discussed earlier, the set $A(t_i)$ can be interpreted in two 
different ways. The first is in terms of the preferences of $t_i$: she is indifferent between schools in $A(t_i)$ and strictly prefers any 
school in $A(t_i)$ to her initial assignment $O(t_i)$, which in turn is strictly preferred to any other deficit school. 
The set $A(t_i)$ can also be interpreted as a compatibility requirement. For instance, $t_i$ could be a primary school teacher and $A(t_i)$ is 
the set of deficit schools with vacancies for primary school teachers. 
Similarly, one can think of reassignment of doctors across public hospitals where the compatibility is determined by a doctor's specialization. 
To avoid trivialities, we assume $A(t_i)$ is non-empty for each $t_i \in T$. 

Every school $s_j\in S$ has a surplus $\alpha_j$ which is a positive integer. Every school $d_k\in D$ 
has a deficit $\beta_k$ which is a positive integer.\footnote{In our basic model, we do not consider 
surplus schools with zero surpluses or deficit schools with zero deficits since they play no role in 
the analysis. In the extension in Section \ref{sec:ext}, we will include surplus schools with zero surpluses.}

We assume that there are $L$ deficit schools, i.e., $|D|=L$. 
The initial deficit vector is denoted by $\beta \equiv (\beta_1,\ldots,\beta_{L})$. Throughout, we will deal with $L$-dimensional (deficit) vectors of the form $x \equiv (x_1,\ldots,x_{L})$. Also for any subset $B \subseteq D$, 
define $x(B) := \sum_{d_k \in B}x_k$.

A {\bf transfer} is a map 
$\sigma:T\rightarrow D\cup S$ satisfying 
the following properties: 

\begin{enumerate}
	\item {\bf No transfer to an unacceptable deficit school:} 
    $$\sigma(t_i)\in A(t_i)\cup \{O(t_i)\}~\qquad~\forall~t_i\in T,$$ 
	
    \item {\bf No surplus school can become a deficit school:}
    $$|\{t_i: O(t_i) = s_j, \sigma(t_i)\neq s_j\}|\leq \alpha_j~\qquad~\forall~s_j\in S,$$
    
	\item {\bf No deficit school can become a surplus school:}
    $$|\{t_i\in T: \sigma(t_i)=d_k\}|\leq \beta_k~\qquad~\forall~d_k\in D.$$ 
\end{enumerate}

The first requirement is an individual rationality constraint for teachers. The second can be interpreted as an individual rationality constraint for 
surplus schools. The third is a natural fairness requirement.

In Section \ref{sec:ext}, we consider extensions in which some of these restrictions are relaxed, and explain why our main result continues to hold.

\subsection{Objectives of the planner}

Every transfer $\sigma$ generates a post-transfer 
deficit vector $\beta^{\sigma}$ where 
\begin{align*}
\beta^{\sigma}_k := \beta_k-|\{t_i\in T:\sigma(t_i)=d_k\}|~\qquad~\forall~d_k \in D. 
\end{align*}

We assume that the planner chooses a transfer by comparing post-transfer deficit vectors. In this respect, our model 
differs from \citet{CTT22} where the primary concern is to improve the welfare of teachers. We believe that our assumption is 
consistent with the provisions of the RTE act --- the planner's goal is to improve deficits while maintaining individual 
rationality for teachers and surplus schools.

A planner may have one of several objectives. A {\it utilitarian} planner wishes to minimize 
the aggregate sum of deficits or transfer as many 
teachers as possible. An {\it egalitarian} planner 
wishes to minimize the worst deficits lexicographically. To formalize this objective, it is helpful to introduce some notation and additional notions, which we turn to next.

Let $\gamma$ be an arbitrary vector in $\Re^{L}$. Let $[\gamma]$ denote the vector where the components of 
$\gamma$ are ordered from highest to lowest, i.e. $[\gamma]=(\gamma_{[1]},\ldots,\gamma_{[L]})$ 
and $\gamma_{[1]}\geq \gamma_{[2]}\ldots \geq \gamma_{[L]}$. Consider $\gamma, \gamma'\in \Re^{L}$. 
We say $\gamma$ lexicographically dominates $\gamma'$ 
if there exists an integer $r\in \{1,\ldots,L\}$ such that $\gamma_{[1]}=\gamma'_{[1]},\ldots,
\gamma_{[r]}=\gamma'_{[r]}$ and $\gamma_{[r+1]}< \gamma'_{[r+1]}$ or $\gamma_{[r]}=\gamma'_{[r]}$
for all $r\in \{1,\ldots,L\}$. An egalitarian planner would like 
to choose a transfer $\sigma$ such that $\beta^{\sigma}$
lexicographically dominates $\beta^{\sigma'}$ for all 
other $\sigma'\in \Sigma$. A utilitarian planner would like to choose a transfer that minimizes $\sum_{k} \beta_k^{\sigma}$ over all $\sigma \in \Sigma$. The set of all possible post-transfer deficit vectors is finite. It follows that the utilitarian and egalitarian solutions exist. 
In general, there is no reason to believe that these solutions 
will coincide.\footnote{\label{fn:utlex} Consider the following example. There 
are two deficit schools and the set of post-transfer deficit 
vectors is the set $\{(\beta_1,\beta_2)\in {\mathbb N}^2: 2\beta_1+\beta_2\geq 30\}$. 
The utilitarian and egalitarian solutions 
would pick $(15,0)$ and $(10,10)$ respectively.}
However, these objectives can be reconciled if we can find a transfer that generates a {\sl Lorenz dominant} deficit vector, which is 
defined next.

Consider $\gamma,\gamma'\in \Re^{L}$. We say 
$\gamma$ Lorenz dominates $\gamma'$ (denoted by $\gamma \lord \gamma'$)
if 
\begin{equation}\label{eq:major}
\sum_{i=1}^{k} \gamma_{[i]} \leq \sum_{i=1}^{k} \gamma'_{[i]}, \; \ \text{for all}\ k\in \{1,\ldots,L\}. 
\end{equation}
In the language of majorization (\cite{Marshall79,Hardy52}), $\gamma$ is weakly majorized by $\gamma'$ if (\ref{eq:major}) holds. If 
Inequality (\ref{eq:major}) for the case $k=L$ holds with equality, we say $\gamma$ is majorized by $\gamma'$. 
Our objective is to find a transfer (if it exists) such that the resulting deficit vector 
is weakly majorized by every deficit vector that can be generated by a transfer. Such a transfer is a Lorenz dominant transfer. 
\begin{defn}
A transfer $\sigma^\star$ is {\bf Lorenz dominant} if $\beta^{\sigma^\star} \lord \beta^{\sigma}$
for every other transfer $\sigma$. 
\end{defn}
There are several equivalent definitions of weak majorization --- see Chapters 2, 3, and 4 of \cite{Marshall79}.
For example, 
\begin{fact}[\citet{HLP29}]
    \label{fact:hlp}
    The vector $\gamma$ Lorenz dominates $\gamma'$ if and only if for all continuous, increasing, and convex functions $g:\Re \rightarrow \Re$, we have 
    $$\sum_{k=1}^L g(\gamma_k) \le \sum_{k=1}^L g(\gamma'_k).$$
\end{fact} 
One can interpret $g$ to be the social planner's ``cost" of deficits.

Applying Fact \ref{fact:hlp}, a transfer $\sigma^\star$ is Lorenz dominant if for every other transfer $\sigma$ and every continuous, increasing, and convex function $g: \Re \rightarrow \Re$, we have 
\begin{align*}
    \sum_{d_k \in D} g(\beta^{\sigma^\star}_k) \le \sum_{d_k \in D} g(\beta^{\sigma}_k)
\end{align*}

The order $\lord$ is a partial order, which implies 
that a Lorenz dominant transfer need not exist (as in the 
example in Footnote \ref{fn:utlex}). If a Lorenz dominant transfer exists, 
it simultaneously reconciles the objectives of utilitarian, egalitarian and 
welfare maximizing planners --- reflected in an appropriate choice of the $g$ function above. 
\citet{BM04} and \citet{bochet2012balancing} prove the existence of a Lorenz dominant solution
in some matching problems that allow continuous or fractional solutions.  
However, ours is a discrete problem: we seek a transfer that generates an integer deficit vector which Lorenz dominates 
every other feasible integer deficit vector. Our main result is that such a transfer exists. 

\begin{theorem}
    \label{theo:main}
    A Lorenz dominant transfer exists.
\end{theorem}

The proof of Theorem \ref{theo:main} constructs a Lorenz dominant transfer using a two-stage algorithm based on a network flow formulation of the problem, which we describe next.

\section{A procedure to find the Lorenz dominant transfer}
\label{sec:network}

In this section, we describe a procedure to find a Lorenz dominant transfer to complete the proof of Theorem \ref{theo:main}. 
This is done in three steps, described in next three subsections. In the first step, we formulate the teacher 
transfer problem as a network flow problem. In the second, we use an existing result of \cite{megiddo1974optimal}, and connect 
the {\sl reduced-form} version of the network flow problem to the core constraints of a convex cooperative game. We then apply 
an algorithm due to \cite{megiddo1974optimal} and \cite{dutta1989concept} (MDR algorithm) to find a Lorenz dominant deficit vector and corresponding 
transfers. The output of the MDR algorithm typically generates fractional transfers. The final step of our procedure consists 
of a careful rounding of the output of the MDR algorithm to obtain a Lorenz dominant transfer. This step involves the use of a standard algorithm (like Ford-Fulkerson) 
to find a maximum flow in an augmented network that is obtained  by modifying the network  of the first step using the output of the MDR algorithm. 

\subsection{A network flow formulation}

The teacher transfer problem is formulated as a {\sl single source multiple sink 
network flow problem}.\footnote{For a comprehensive survey of network flows see \citet{ahuja1988network}.}
We consider two versions of the problem, one where flows are restricted to be non-negative
integers and the other where flows are allowed to be non-negative real numbers.

Both versions use the same network structure that we describe below. 
\begin{itemize}
    \item The set of nodes $N$ in the graph consists of:
    \begin{enumerate}
        \item a source node $N_0$
        \item a node for each surplus school in $S$
        \item a node for each teacher in $T$
        \item a node for each deficit school in $D$: we refer to these nodes as {\sl sink} nodes
    \end{enumerate}
    Thus, $N \equiv \{N_0\}\cup S\cup T\cup D$. 

    \item The set of edges $E$ and their capacities are as follows: 
    \begin{enumerate}
        \item a directed edge $(N_0,s_j)$ for every surplus school $s_j \in S$, with capacity $\alpha_j$
        \item a directed edge $(O(t_i),t_i)$ for every teacher $t_i \in T$, with capacity $1$
        \item a directed edge $(t_i,d_k)$ for every teacher $t_i \in T$ and every deficit school $d_k \in A(t_i)$, with capacity $1$
    \end{enumerate}
    \item Each deficit school (sink) node $d_k \in D$ has a capacity $\beta_k$.
    \footnote{Note that we can convert node capacities to edge capacities by creating a copy of each sink node and adding a directed edge of capacity $\beta_k$ from sink node $d_k$ to its copy.}
\end{itemize}
We let $G$ denote this network.

A flow is a map $f: E \rightarrow \Re_+$ satisfying {\sl flow balancing}:

\begin{enumerate}
    \item $f(N_0,s_j)=\sum \limits_{t_i: s_j = O(t_i)} f(s_j,t_i)$ for each surplus school $s_j\in S$,
	
	\item $f(O(t_i),t_i)=\sum \limits_{d_k\in A(t_i)} f(t_i,d_k)$ for each teacher $t_i\in T$,
\end{enumerate}
and {\sl capacity} constraints:
\begin{enumerate}
    \item $f(N_0,s_j) \le \alpha_j$ for all $s_j\in S$,

    \item $f(s_j,t_i) \le 1$ for all edges $(s_j,t_i)$ such that $s_j=O(t_i)$,

    \item $f(t_i,d_k) \le 1$ for all edges $(t_i,d_k)$ such that $d_k \in A(t_i)$,
	
	\item $\sum \limits_{t_i\in T} f(t_i,d_k) \le \beta_k$ for all $d_k\in D$.
\end{enumerate}

We say a flow is an {\sl integer flow} if the flow on every edge of the network is a non-negative integer.

For every integer flow $f$ in $G$, we
define a corresponding transfer $\sigma$ as follows:
\begin{align*}
\sigma(t_i) & =
\begin{cases}
d_k, & \mbox{if }\  f(t_i,d_k)=1\\
O(t_i), & \mbox{if  }\   f(t_i,d_k)=0\ \mbox{for all}\ d_k\in A(t_i)
\end{cases}
\end{align*}

It is easy to verify that $\sigma$ satisfies the feasibility properties in the definition of a transfer. Similarly,
given a transfer $\sigma$, we define 
a corresponding {\bf integer} flow $f$ as follows: For the edge
$(t_i,d_k)$, let
\begin{align*}
f(t_i,d_k) & =
\begin{cases}
1, & \mbox{if }\  \sigma(t_i)=d_k\\
0, & \mbox{if  }\  \sigma(t_i)=O(t_i)
\end{cases}
\end{align*}
For the edge $(O(t_i),t_i)$, let
\begin{align*}
f(O(t_i),t_i) & =
\begin{cases}
1, & \mbox{if }\  \sigma(t_i)=d_k\\
0, & \mbox{if  }\  \sigma(t_i)=O(t_i)
\end{cases}
\end{align*}
For the edge $(N_0,s_j)$, let
\begin{equation*}
f(N_0,s_j)=|\{t_i: s_j=O(t_i), \sigma(t_i)\neq s_j\}|
\end{equation*}
Once again it is easy to verify that $f$ is integer-valued and satisfies flow balancing and capacity constraints. It follows that there is a bijection between transfers and the set of integer flows in $G$. For a flow $f$ in $G$, we denote the flow into a sink node (deficit school) $d_k$ by $f(d_k)$.

A flow $f$ is a {\bf maximum flow} if for every other flow $f'$ we have 
\begin{align*}
    \sum_{d_k \in D} f(d_k) &\ge \sum_{d_k \in D} f'(d_k)
\end{align*}
It is well-known that a maximum integer flow exists whenever the capacities of nodes and edges are integral. Consequently, there is a maximum flow that corresponds to a transfer. By Lorenz dominance condition (\ref{eq:major}), the following observation is immediate. 
\begin{obs}
    \label{obs:maxflow}
    A Lorenz dominant transfer (if it exists) corresponds to an integral maximum flow.
\end{obs}

\subsection{An illustrative example}
We present an example to highlight some important features of our model.

There are two surplus schools $\{s_1,s_2\}$ each with a surplus of one teacher. 
Teacher $t_1$ is initially assigned to surplus school $s_1$; teachers $t_2$ and $t_3$ are initially assigned to surplus school $s_2$.
Surplus school $s_2$ can transfer at most transfer one teacher. There are three deficit schools $\{d_1,d_2,d_3\}$ with deficits $\beta_1=1, \beta_2=3, \beta_3=2$. Teachers $t_1$ and $t_2$ find $d_1$ and $d_2$ acceptable, while teacher $t_3$ finds all three deficit schools acceptable. The corresponding network is shown in Figure \ref{fig:second}.
\begin{figure}
    \centering
    \includegraphics[width=0.5\linewidth]{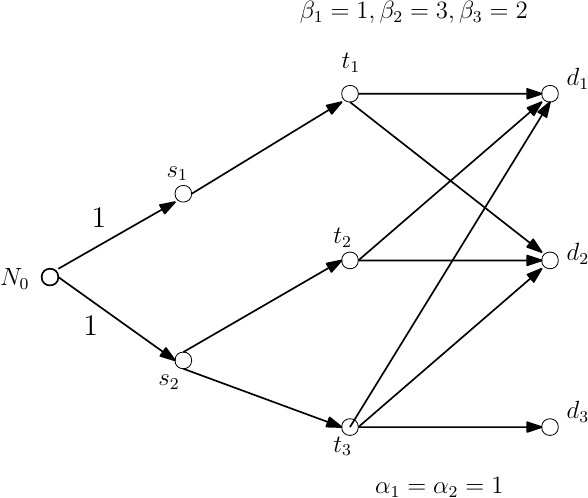}
    \caption{Example $1$}
    \label{fig:second}
\end{figure}

As we will see later, the {\sl fractional} Lorenz dominant solution is the deficit vector $(1,\frac{3}{2},\frac{3}{2})$. This corresponds to a transfer of no teachers to school $d_1$, $\frac{3}{2}$ teachers to school $d_2$ and $\frac{1}{2}$ teacher to school $d_3$. This transfer equalizes the deficits of schools $d_2$ and $d_3$ without changing the deficit of school $d_1$.

Our procedure will result in an integer deficit vector $(1,1,2)$ with $d_1$ and $d_2$ having a (post-transfer) deficit of $1$, while the deficit of $d_3$ remains unchanged at $2$. This can be achieved by transferring teachers $t_1$, and either $t_2$ or $t_3$ (but not both), to school $d_2$; the remaining teacher is not transferred out of school $s_2$ in order to respect the surplus constraint of $s_2$.

The vector $(1,1,2)$ Lorenz-dominates all other possible integer post-transfer deficit vectors. 

We wish to emphasize some features of our model:
\begin{itemize}
    \item {\sl Surplus constraints are important.} If we ignored surplus constraints in Example  \ref{fig:second},  all the three teachers could be transferred. This would reduce deficits further, but would also turn a surplus school into a deficit school, violating a natural participation constraint for the surplus school. In our model the set of teachers is partitioned according to the schools they are currently assigned to.  There is an upper bound on the number of teachers that can be transferred from each element of the partition. This constraint distinguishes our model from standard matching problems. 

    \item {\sl Role of existing deficits.} In Example $1$, the Lorenz dominant solution does not transfer any teacher to $d_1$ since its existing deficit is the lowest amongst the three deficit schools. On the other hand, if the deficit in $d_1$ were high, say $5$, then two teachers would be transferred to $d_1$ to equalize deficits. 

    \item {\sl Many transfer-maximizing solutions.} There are many transfers that maximize the total transfer to deficit schools. For instance, transferring $t_1$ to $d_1$ and transferring $t_2$ to $d_2$ makes the maximal ($2$) number of transfers resulting in the deficit vector $(0,2,2)$ but that is clearly Lorenz dominated by our solution.

\end{itemize}

\subsection{The reduced form and a cooperative game}

In this subsection, we present a result in \citet{megiddo1974optimal}  that allows us to restrict attention {\it solely} to flows into deficit schools. According to the result, a flow to deficit schools can arise from a feasible flow in network $G$ if and only if it satisfies some conditions. These conditions appear in the form of the core constraints of a suitably constructed cooperative game. This reformulation of the problem in {\sl reduced-form} simplifies our analysis.

In order to describe the reduced-form, we introduce some more concepts. 
For every $B \subseteq D$, a flow $f$ is a {\bf $B$-maximum flow} if for every other flow $f'$ we have 
\begin{align*}
    \sum_{d_k \in B}f(d_k) &\ge \sum_{d_k \in B}f'(d_k)
\end{align*}
Let $v(B):= \sum_{d_k \in B}f(d_k)$, i.e., the maximum flow into sinks in $B$. An integer $B$-maximum flow exists, and hence, $v(B)$ is a non-negative integer for every $B \subseteq D$. 

Define $w:2^D\rightarrow \Re$ as follows:
\begin{align*}
w(B):=\beta(B)-v(B)~\qquad~\forall~B \subseteq D. 
\end{align*}

The pair $\langle D,w\rangle$ is a cooperative 
game, where $D$ is the set of players and $w$ is
an integer-valued characteristic function. Since 
$D$ will be held fixed throughout the analysis, 
we will refer to this game simply by $w$. 

Let $h=(h_1,\ldots,h_L)$ be an $L$-dimensional vector of real numbers. 
The vector $h$ is {\bf achievable} if there is a flow $f$ in $G$ such that $h_k = \beta_k - f(d_k)$ for all $d_k \in D$. In other words, $h$ is the post-transfer vector generated by some feasible flow. If $h$ is an integral vector, it must be generated by an integral flow vector, i.e., a transfer. 
\begin{prop}[\cite{megiddo1974optimal}]
    \label{prop:megiddo}
    A vector $h \equiv (h_1,\ldots,h_L)$ is achievable if and only if 
    \begin{align}\label{eq:core}
        h(B) &\ge w(B)~\qquad~\forall~B \subseteq D
    \end{align}
    
    Further, $w$ is convex (supermodular) and integer-valued, i.e., for any $A \subsetneq B \subseteq D$ and $d_k \notin B$,
    \begin{align}
        w(B \cup \{d_k\}) - w(B) \ge w(A \cup \{d_k\}) - w(A)
    \end{align}
\end{prop}
\begin{proof}
    The proof follows from a straightforward application of results in \cite{megiddo1974optimal}. 
     Let $x \equiv (x_1,\ldots,x_L)$ be any non-negative vector. Lemma 4.1 in \cite{megiddo1974optimal} shows that $x$ can be generated from a flow if and only if $x(B) \le v(B)$ for all $B \subseteq D$. This immediately implies (\ref{eq:core}).
    
    Lemma 3.2 in \cite{megiddo1974optimal} shows that $v$ is concave. Since $\beta$ is linear and $-v$ is convex, it follows that $w$ is convex.
\end{proof}

As a consequence of Proposition \ref{prop:megiddo}, we can restrict attention to a reduced form problem of analyzing post-transfer deficit vectors. Moreover, achievability consists of verifying a family of subset inequalities. Since $w$ is convex, these inequalities closely resemble the core constraints of a convex game, 
with one important difference: equality for $B=D$ in (\ref{eq:core}) is not imposed.

We refer to the constraints in (\ref{eq:core}) as {\bf relaxed core} constraints.\footnote{In discrete optimization, these constraints are variants of the polymatroid constraints.}

We can summarize the results thus far, as 
follows. The teacher transfer problem can be reformulated as 
a network flow problem where integer flows correspond to transfers. 
A convex cooperative game $w$ is induced by the
network flow problem, where deficit schools are the players. 
For every coalition $B\subseteq D$, $w(B)$ is 
an integer. In addition, a real 
deficit vector is achievable in the network flow problem 
(where fractional teacher flows are allowed) if 
and only if it belongs to the relaxed core of $w$. 

\subsection{The Megiddo-Dutta-Ray algorithm}

The first stage of our algorithm consists of applying a result from  cooperative game theory  to find a {\sl fractional} Lorenz dominant transfer. \citet{dutta1989concept} show that the core of a convex game always contains a Lorenz dominant 
allocation and provide an algorithm to identify it.\footnote{\citet{megiddo1974optimal} proposed 
the same algorithm to construct a lexicographically optimal 
flow. Note that a lexicographically optimal solution of a problem always exists but a Lorenz dominant solution may not.} Since our achievability constraints are {\sl relaxed} rather than standard core constraints, we cannot directly apply their result. However, as we show in Appendix \ref{app:mdr}, their algorithm also produces a fractional Lorenz dominant allocation with relaxed core constraints.

As noted earlier, a fractional solution does not correspond to a transfer. In the second stage of our algorithm we {\sl round} the output of the MDR algorithm. This produces an integral deficit vector that corresponds to a transfer and one  that Lorenz dominates all other deficit vectors arising from transfers. Details of the second stage of the algorithm are in the next section.

The MDR algorithm inductively partitions the set of deficit schools $D$ into sets $(D_1,\ldots,D_{k^\star})$. For any $i \in \{1,\ldots,k^\star\}$, let $\overline{D}_i = \cup_{j=1}^i D_j$. Define $D_1$ (and $\overline{D}_1)$ as the set having the highest average worth:
\begin{align*}
    D_1 \in \arg \max_{B \subseteq D} \frac{w(B)}{|B|}
\end{align*}
Given $D_1,\ldots,D_k$ for some $k$, define $D_{k+1}$ as the set having the highest average marginal contribution to $\overline{D}_k$:
\begin{align*}
    D_{k+1} \in \arg \max_{B \subseteq D \setminus \overline{D}_k}\frac{w(B \cup \overline{D}_k) - w(\overline{D}_k)}{|B|}
\end{align*}

Ties in the determination of $D_{k+1}$ are broken arbitrarily.

Since $D$ is finite, the algorithm generates a sequence of disjoint subsets of deficit schools: $(D_1,\ldots,D_{k^\star})$. Using this sequence, we construct a deficit vector $h^\star$ as follows. For every $D_j \in \{D_1,\ldots,D_{k^\star}\}$,
\begin{align}\label{eq:mdrdef}
    h^\star_i := \frac{w(\overline{D}_j) - w(\overline{D}_{j-1})}{|D_j|}~\qquad~\forall~d_i \in D_j,
\end{align}
where $\overline{D}_0 \equiv \emptyset$ with $w(\emptyset)=0$.

Theorem \ref{theo:mdr} in Appendix \ref{app:mdr} shows that $h^\star$ is Lorenz dominant deficit vector. Its proof  relies on some key claims. Claim \ref{cl:maxach} shows that $h^\star$ is achievable. Claim \ref{cl:orderld} establishes an ordering property
of $h^*$:   $h^\star_i \ge h^\star_j$ if $d_i \in D_k$ and  $d_j \in D_{k+1}$. Another important property of $h^{*}$ that follows from  (\ref{eq:mdrdef}) is
\begin{align} \label{eq:corebind}
    \sum_{d_i \in \overline{D}_j} h^\star_i = w(\overline{D}_j)~\qquad~\forall~j \in \{1,\ldots,k^\star\} 
\end{align}
The sets $\overline{D}_j$, $j=1, \ldots k^*$ are the subsets of deficit schools where the core constraints are binding for $h^\star$. 

According to (\ref{eq:mdrdef}), $h^\star_i$ is a ratio of two integers. This immediately suggests that $h^\star$ is not integral and that rounding is required in order to obtain a meaningful result on teacher transfers. The next two subsections address this issue.

\subsection{Arbitrary roundings do not work: An example}

This subsection provides an example (Example $2$) to show that the rounding issue is a delicate one. The example also illustrates the working
of the MDR algorithm.

The sets of deficit schools, surplus schools and teachers are 
$S =\{s_1,s_2,s_3\}$, $D=\{d_1,d_2,d_3,d_4,d_5,d_5,d_7\}$ and $T=\{t_1,t_2,t_3,t_4,t_5\}$ respectively. 
The initial assignment of teachers is given by: $O(t_1)=s_1$, $O(t_2)=O(t_3)=s_2$ and $O(t_4)=O(t_5)=s_3$. 
The acceptable sets are as follows: $A(t_1)=\{d_1,d_2\}$, $A(t_2)=A(t_3)=\{d_1,d_2,d_3,d_4,d_5\}$, 
$A(t_4)=\{d_6\}$ and $A(t_5)=\{d_6,d_7\}$. The deficit of every school is $5$. 
The surplus of schools $s_1$, $s_2$ and $s_3$ are $\alpha_1=1$ and $\alpha_2=\alpha_3=2$ respectively. 
The example and the associated network flow formulation of the problem
is illustrated in Figure \ref{fig:example}. 

\begin{figure}[!hbt]
	\centering
	\includegraphics[width=2.5in]{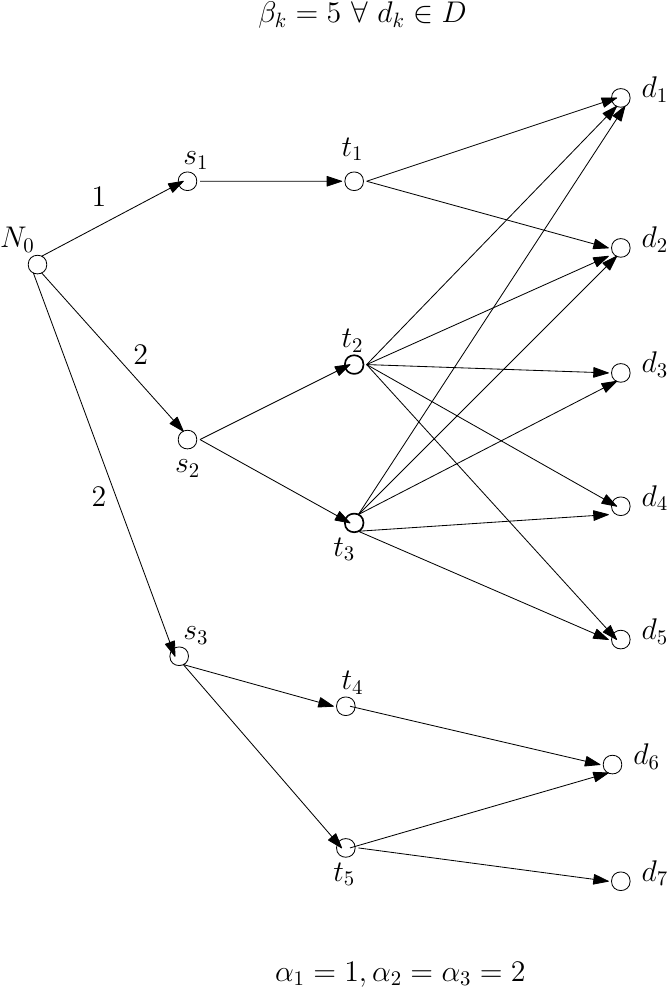}
	\caption{Example $2$}
	\label{fig:example}
\end{figure}

We omit the details of the computation of the characteristic function $w$  and simply note that the maximum average worth of a coalition is achieved by coalition $D_1=\{d_1,d_2,d_3,d_4,d_5\}$.
\begin{align*}
    \frac{w(D_1)}{|D_1|} = \frac{22}{5} = \max_{B: B \subseteq D} \frac{w(B)}{|B|}.
\end{align*}
This generates $D_1$ of the MDR algorithm. Having identified $D_1$, the average marginal contributions of subsets of remaining deficit schools are computed.
\begin{align*}
w(D_1 \cup \{d_6\}) - w(D_1) &= 3 \\
w(D_1 \cup \{d_7\}) - w(D_1) &= 4 \\
\frac{w(D_1 \cup \{d_6,d_7\}) - w(D_1)}{2} &= 4
\end{align*}
There are two possible choices for $D_2$: either $D_2 =\{d_7\}$ or $D_2=\{d_6,d_7\}$. 
Either choice can be made, so we set  $D_2=\{d_6,d_7\}$. The partition 
chosen by the MDR algorithm is therefore $\{\{d_1,\ldots,d_5\},\{d_6,d_7\}\}$.  The resultant Lorenz dominant deficit vector is:
\begin{align*}
    h^\star_i = 
    \begin{cases}
        4.4 & \textrm{if}~d_i \in D_1 \\
        4 & \textrm{if}~d_i \in D_2
    \end{cases}
\end{align*}
 It is clear that the vector $h^\star$ cannot be obtained from an integer flow (or transfer). 

An important observation here is that $h^\star$ cannot be approximated by an arbitrary rounding.
For example, consider the following rounding of $h^\star$: $$\hat{h}_1=\hat{h}_2=5,\hat{h}_3=\hat{h}_4=\hat{h}_5=4, \hat{h}_6=\hat{h}_7=4.$$
This is a ``valid" rounding of $h^\star$ in the sense that for each $d_i \in D$, we have
\begin{align*}
    \hat{h}_i \in \{ \lfloor h^\star_i \rfloor, \lceil h^\star_i \rceil\}.
\end{align*}
Further, $h^\star(D_1)=\hat{h}(D_1)$ and $h^\star(D_2)=\hat{h}(D_2)$.\footnote{A rounding of $h^\star$ satisfying these two properties will be referred to as an MDR-consistent rounding. See Section \ref{sec:rounding}.}

The difficulty is that $\hat{h}$ is {\bf not achievable}, i.e., there is no feasible transfer that can generate this deficit vector. To see this, consider $B=\{d_3,d_4,d_5\}$ and verify that $w(B)=13$. But $\hat{h}(B)=12 < w(B)$ implying non-achievability.
This can also be directly verified by
noting that only two teachers $t_2$ and $t_3$ can be transferred to $\{d_3,d_4,d_5\}$. Consequently, the aggregate
deficit of $\{d_3,d_4,d_5\}$ cannot be reduced from $15$ to $12$.

\subsection{MDR-consistent rounding}
\label{sec:rounding}

The previous example demonstrates that an arbitrary rounding of the output of the MDR algorithm may not be achievable. The critical step in the proof of Theorem \ref{theo:main} is to show that there exists an achievable rounding where
the core binding constraints (\ref{eq:corebind}) continue to hold for the new integer deficit vector. The proof is completed by showing that this vector Lorenz dominates every other integral achievable deficit vector.

\begin{defn}
    \label{def:mdrcons}
    A deficit vector $h$ is {\bf MDR-consistent} if it is an integral achievable deficit vector such that
    \begin{align}\label{eq:mdr1}
        h_i &\in \{ \lfloor h^\star_i \rfloor, \lceil h^\star_i \rceil\}~\qquad~\forall~d_i \in D \\
        h(\overline{D}_j) &= w(\overline{D}_j)~\qquad~\forall~j \in \{1,\ldots,k^\star\} \label{eq:mdr2}
    \end{align}
\end{defn}

An MDR-consistent deficit vector corresponds to a transfer as it is integral and achievable. 
\begin{prop}[Rounding]
\label{prop:int}
An MDR-consistent deficit vector exists.
\end{prop}
\begin{proof}
An auxiliary network $G^{{\rm AU}}$ is constructed by augmenting the network $G$ using the output of the MDR algorithm $\{D_1,\ldots,D_{k^\star}\}$. Abusing notation slightly, the deficit of any school in $D_j$ generated by the MDR algorithm is denoted  by $h^\star_j$. Though $h^\star_j$ is not integral, $h^\star_j |D_j|$ is integral and equals $w(\overline{D}_j) - w(\overline{D}_{j-1})$. The following set of nodes and edges is added to the network $G$ to construct $G^{{\rm AU}}$:
\begin{enumerate}
    \item For each $j \in \{1,\ldots,k^\star\}$, a node is added representing the subset of deficit schools $D_j$. Each node $D_j$ has a capacity of $\sum_{d_i \in D_j}\big(\beta_i - h^\star_i)$ (integral).
    
    \item For each $j \in \{1,\ldots,k^\star\}$ and each deficit school $d_i \in D_j$, an edge $(d_i,D_j)$ is added. The flow in each such edge must lie in $[\lfloor \beta_i - h^\star_i\rfloor, \lceil \beta_i - h^\star_i \rceil ]$.
\end{enumerate}
All capacities and bounds on flows in $G^{{\rm AU}}$ are integers.
The construction of $G^{{\rm AU}}$ for the example in Figure \ref{fig:example} is illustrated in Figure \ref{fig:aux}. 

\begin{figure}[!hbt]
	\centering
	\includegraphics[width=4in]{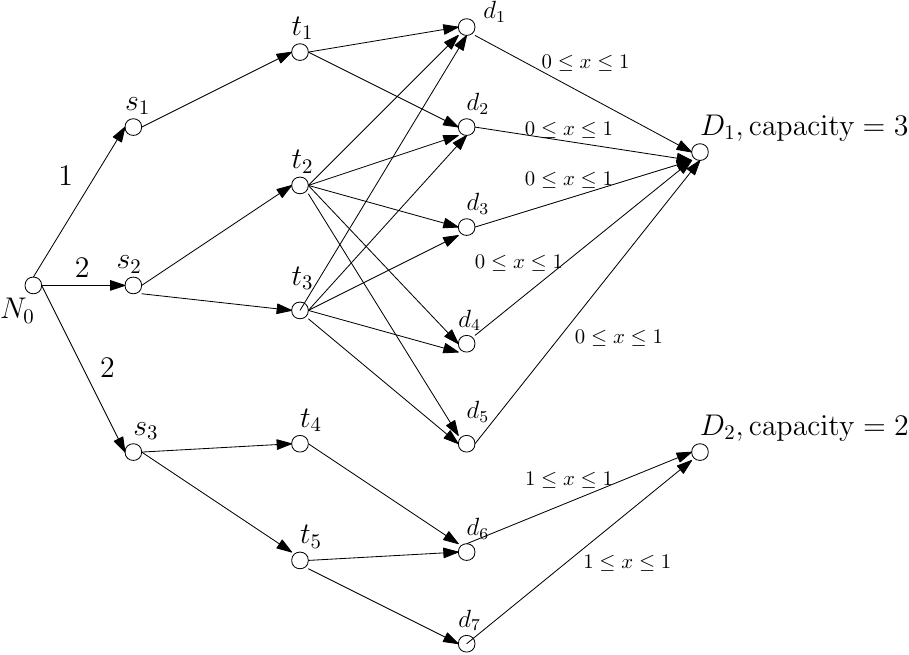}
	\caption{Graph $G^{{\rm AU}}$ in Example in Figure \ref{fig:example}}
	\label{fig:aux}
\end{figure}

Since $h^\star$ is achievable, there 
exists flow $f^\star$ such that $f^\star(d_k)=\beta_k-h^\star_k$ for
all $d_k \in D$. Construct the flow $f^{\star \star}$ by augmenting $f^\star$ as follows: for 
every $d_k$, the flow on the edge $(d_k,D_j)$ 
where $d_k \in D_j$ is $\beta_k-h^\star_k$. Clearly $f^{\star \star}$ satisfies the additional flow balancing and capacity constraints in
$G^{{\rm AU}}$, and hence, is a feasible flow in the augmented network. The flow received by a (sink) node $D_j$ in 
$f^{\star \star}$ is, 
\begin{align}
\label{eq:hath0}
    f^{\star \star}(D_j) =  \sum_{d_k\in D_j} f^{\star \star}(d_k,D_j) = \sum_{d_k \in D_j} \big(\beta_k-h^\star_k\big) = \beta(D_j) - h^\star(D_j)
\end{align}

The RHS of the equation above is the capacity of node $D_j$. 
It follows that $f^{\star \star}$ is a maximum flow in $G^{{\rm AU}}$.

Since all capacities in $G^{{\rm AU}}$ are integers, the integrality 
theorem in network flows (see Page 186 in \citet{ahuja1988network})
implies that there exists a integral maximum flow $\hat f$. Therefore 
$\sum_{j=1}^{k^\star} f^{\star \star} (D_j)=\sum_{j=1}^{k^\star} \hat f(D_j)$. 
In order to satisfy capacity constraints of each node $D_j$, 
we must have 
\begin{align}\label{eq:bind1}
    f^{\star \star}(D_j) &= \hat f(D_j)~\qquad~\forall~j \in \{1,\ldots,k^\star\}. 
\end{align}

Define a new deficit vector $\hat h_k:=\beta_k-\hat f(d_k)$ for all $d_k\in D$. Hence, 
\begin{align}
    \hat{h}(D_j) = \beta(D_j) - \hat{f}(D_j)~\qquad~\forall~j \in \{1,\ldots,k^\star\} \label{eq:hath1}
\end{align}
Note that $\hat h$ is the post-transfer deficit vector corresponding 
to $\hat f$. From (\ref{eq:hath0}), (\ref{eq:bind1}), and (\ref{eq:hath1}), we get  $\hat{h}(D_j)=h^\star(D_j)$ for all $j \in \{1,\ldots,k^\star\}$.
This further implies that 
\begin{align}
    \label{eq:bind3}
    \hat{h}(\overline{D}_j) &= h^\star(\overline{D}_j) = w(\overline{D}_j)~\qquad~\forall~j \in \{1,\ldots,k^\star\},
\end{align}
where the last equality follows from the property we established for $h^\star$ (see Equation (\ref{eq:corebind})).

We complete the proof by showing that the integrality constraint is satisfied. By the construction of $G^{{\rm AU}}$, $\hat f$ induces an integer flow in the original network $G$.
It follows that $\hat h$ is an integral achievable post-transfer deficit 
vector in $G$.

For any school $d_k\in D$, we have $\hat f(d_k)=\beta_k-\hat h_k$. 
By flow balancing at node $d_k$ where $d_k\in D_j$, we have 
$\hat f(d_k)=\hat f(d_k,D_j)$. Since the flow $\hat{f}(d_k,D_j)$ must satisfy the capacity constraint of edge $(d_k,D_j)$, this implies
\begin{align*}
    \left \lfloor{\beta_k-h^\star_k}\right \rfloor &\leq \hat f(d_k) \leq \left \lceil{\beta_k-h^\star_k}\right \rceil\\
	\implies  \left \lfloor{\beta_k-h^\star_k}\right \rfloor &\leq \beta_k-\hat h_k \leq \left \lceil{\beta_k-h^\star_k}\right \rceil \\
	\implies  \beta_k-\left \lceil{h^\star_k}\right \rceil &\leq \beta_k-\hat h_k \leq \beta_k-\left \lfloor{h^\star_k}\right \rfloor \\
	\implies  \left \lfloor{h^\star_k}\right \rfloor &\leq \hat h_k \leq \left \lceil{h^\star_k}\right \rceil
\end{align*}
This establishes that $\hat{h}$ is MDR-consistent.
\end{proof}

Proposition \ref{prop:int} together with Proposition \ref{prop:intld} below complete 
the proof of Theorem \ref{theo:main}.
An MDR consistent deficit vector in Example $2$ is: $$\hat{h}_1=5,\hat{h}_2=\hat{h}_3=\hat{h}_4=4,\hat{h}_5=5, \hat{h}_6=\hat{h}_7=4.$$

\begin{prop}[Lorenz dominance]
    \label{prop:intld}
    An MDR-consistent deficit vector Lorenz dominates every other integral achievable deficit vector.
\end{prop}
\begin{proof}
    Consider an MDR-consistent deficit vector $\hat{h}$ --- by definition, it is achievable and integral. Recall that the output of the MDR algorithm is the achievable deficit vector $h^\star$ (not necessarily integral) and a partition of deficit schools $(D_1,\ldots,D_{k^\star})$ satisfying the following: for every $k \in \{1,\ldots,k^\star\}$ and every $d_i \in D_k$,
    \begin{align*}
        h^\star_i &= \frac{w(\overline{D}_k) - w(\overline{D}_{k-1})}{|D_k|} \\
        h^\star(\overline{D}_k) &= w(\overline{D}_k).
    \end{align*}
    MDR-consistency implies $h^\star(\overline{D}_k) = \hat{h}(\overline{D}_k) = w(\overline{D}_k)$ and $\hat{h}_i \in \{ \lfloor h^\star_i \rfloor, \lceil h^\star_i \rceil\}$. There may exist sets $D_k, D_{k+1}$ such that $\{ \lfloor h^\star_i \rfloor, \lceil h^\star_i \rceil\} = \{ \lfloor h^\star_j \rfloor, \lceil h^\star_j \rceil\}$, where $d_i \in D_k$ and $d_j \in D_{k+1}$. Sets whose integral floor and ceilings are the same in $h^\star$  are merged to form a new partition $(D^+_1,\ldots,D^+_{\ell^\star})$.   
    Clearly, $\ell^\star \le k^\star$. Different schools in $D^+_k$ may have different deficits in $h^\star$. However, there are three important facts about the new partition that we note below. The first is that $\{ \lfloor h^\star_i \rfloor, \lceil h^\star_i \rceil\} = \{ \lfloor h^\star_j \rfloor, \lceil h^\star_j \rceil\}$    
    for any $d_i,d_j \in D^+_k$, The second is that,  for any $k \in \{1,\ldots,\ell^\star\}$, we have 
    \begin{align}\label{eq:hlpin1}
        \hat{h}(\overline{D}^+_k) = h^\star(\overline{D}^+_k) = w(\overline{D}^+_k).
    \end{align}
    Finally, we have $\hat{h}_i \ge \hat{h}_j$ for any $d_i \in D^+_k$ and $d_j \in D^+_{k+1}$.

    Assume without loss of generality that $\hat{h}$ is ordered from highest to lowest in each $D^+_k$.
    Let $h$ be an arbitrary achievable integer deficit vector. Pick an arbitrary $j \in \{1,\ldots,L\}$. 
    The proof is completed by showing $\sum_{i=1}^j \hat{h}_{[i]} \le \sum_{i=1}^j h_{[i]}$. 
    We consider two cases.

    \noindent {\sc Case A:} $j = |\overline{D}^+_{\ell}|$ for some $\ell$. From (\ref{eq:hlpin1}), we have 
    \begin{align*}
        \sum_{i=1}^j \hat{h}_{[i]} = \hat{h}(\overline{D}^+_\ell) =_{(a)} w(\overline{D}^+_\ell) \le_{(b)} h(\overline{D}^+_\ell) \le_{(c)} \sum_{i=1}^j h_{[i]}.
    \end{align*}
    Equality (a) follows from MDR-consistency of $\hat{h}$. Inequality (b) follows from achievability of $h$. Inequality (c) follows because the RHS is the sum of the $j$ highest deficits in $h$.

    \noindent {\sc Case B:} Case A does not hold, i.e. $j \ne |\overline{D}^+_\ell|$ for any $\ell$. 
    Suppose the $j$-th highest deficit school in $\hat{h}$ belongs to $D_k^+$. Let $|\overline{D}^+_{k-1}| \equiv \ell_{k-1}$ and $|\overline{D}^+_{k}| \equiv \ell_k$ \footnote{If $k=1$, let $\overline{D}_{k-1}^+=\emptyset$ and $\ell_{k-1}=0$.}, so that $|D^+_k| = \ell_k - \ell_{k-1}$. From Case A, the Lorenz domination inequality holds for $\ell_{k-1}$ and $\ell_k$:
    \begin{align}
        \sum_{i=1}^{\ell_{k-1}}h_{[i]} &\ge \sum_{i=1}^{\ell_{k-1}}\hat{h}_{[i]} = \hat{h}(\overline{D}^+_{k-1}) \label{eq:hhh1} \\
        \sum_{i=1}^{\ell_k}h_{[i]} &\ge \sum_{i=1}^{\ell_{k}}\hat{h}_{[i]} = \hat{h}(\overline{D}^+_k). \label{eq:hh2}
    \end{align}
    
    Since $|D_k^+|=\ell_k - \ell_{k-1}$, the deficit vector $\hat{h}$ restricted to $D_k^+$ is an $(\ell_k - \ell_{k-1})$-dimensional vector of integers that are either the floor or the ceiling of deficit vectors in $h^\star$ restricted to $D_k^+$. Denote the floor and ceiling of $h^\star$ vector in $D_k^+$ by $\overline{h}^\star$ and $\underline{h}^\star$ respectively. It follows that $\hat{h}$ restricted to $D_k^+$ is a vector of the following form:
    \begin{align*}
        (\underbrace{\overline{h}^\star, \ldots, \overline{h}^\star},\underbrace{\underline{h}^\star,\ldots,\underline{h}^\star}).
    \end{align*}

Figure \ref{fig:proof_int} shows the distribution of the deficit vectors $\hat{h}$ and $h$ for schools in $D_k^+$ in an illustrative case. Let $D_k^+$ contain $10$ schools. The vectors $\hat{h}$ and $h$ (restricted to $D_k^+$) are ordered from high to low. Both $\hat{h}$ and $h$ are integral, so that a (deficit, school) pair 
in each vector is a point on the grid in the figure. The vector $\hat{h}$ either takes the floor $\underline{h}^\star$ or the ceiling $\overline{h}^\star$ value
and is represented by {\sc red} squares. The seven schools with highest deficits have value $\overline{h}^\star$ while the remaining three schools have value $\underline{h}^\star$. The vector $h$ is represented by {\sc blue} crosses. The five highest deficits in $h$ are at least as high as $\overline{h}^\star$, the 
next two coincide with $\underline{h}^\star$ and the remaining three are below $\underline{h}^\star$. A crucial observation is the following: since $\hat{h}$ is integral and $\overline{h}^\star$ and $\underline{h}^\star$ are consecutive integers, $h$ values are less than or equal to $\underline{h}^\star$
from the sixth school onwards. In other words, $\hat{h}$ and $h$ satisfy a {\sl single-crossing} property.\footnote{The single-crossing property is trivially achieved in the continuous version, and exploited in Claim \ref{cl:hlp} due to \citet{HLP29}.} This allows us to break the proof of Case B into two parts. 
The first part considers $j$ where the all schools upto the $j$-th highest school have deficit values in $h$ greater than or equal to $\overline{h}^\star$.
In the figure, this corresponds to $j\leq 5$. The second part considers the case where the deficit of the $j$-th highest school in $h$ is $\underline{h}^\star$
or lower, i.e. $j\geq 6$ in the figure. 

    Pick any $j \in \{\ell_{k-1}+1,\ell_{k-1}+2,\ldots,\ell_k\}$. If $h_{[j]} \ge \overline{h}^\star$, then $h_{[\ell_{k-1}+1]} \ge \ldots \ge h_{[j]} \ge \overline{h}^\star$ implies 
    \begin{align*}
        \sum_{i=\ell_{k-1}+1}^j h_{[i]} \ge \sum_{i=\ell_{k-1}+1}^j\overline{h}^\star \ge \sum_{i=\ell_{k-1}+1}^j\hat{h}_{[i]},
    \end{align*}
    where the second inequality uses the fact that $\hat{h}_{i} \le \overline{h}^\star$ for all $i \in \{\ell_{k-1}+1,\ldots,\ell_k\}$.
    Now, using inequality (\ref{eq:hhh1}) and adding $\sum_{i=1}^{\ell_{k-1}}h_{[i]}$ on the LHS and $\hat{h}(\overline{D}^+_{k-1})$ on the RHS, we get 
    \begin{align*}
        \sum_{i=1}^j h_{[i]} \ge \sum_{i=1}^j\hat{h}_{[i]}.
    \end{align*}

    \begin{figure}
        \centering
        \includegraphics[width=0.5\linewidth]{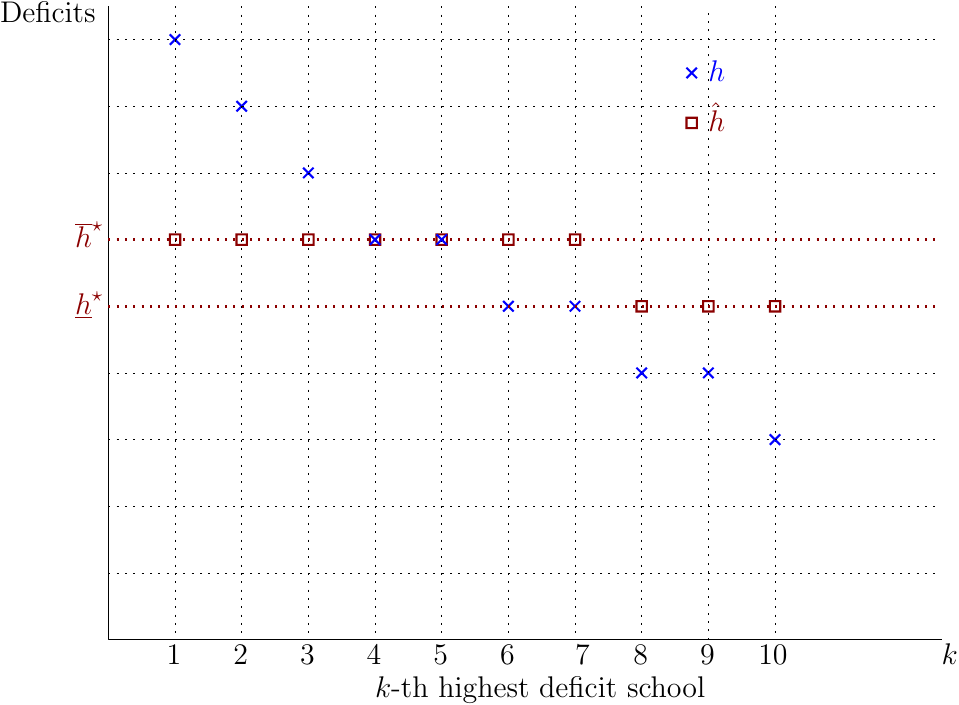}
        \caption{Single crossing of ordered integral vectors $\hat{h}$ and $h$.}
        \label{fig:proof_int}
    \end{figure}
    
    Suppose $h_{[j]} < \overline{h}^\star$. Since $h_{[j]}$ is an integer and $\overline{h}^\star=\underline{h}^\star+1$, we have $h_{[j]} \le \overline{h}^\star - 1 = \underline{h}^\star$. But then, $\underline{h}^\star \ge h_{[j]} \ge h_{[j+1]} \ge \ldots \ge h_{[\ell_k]}$. Hence,
    \begin{align*}
        \sum_{i=1}^j h_{[i]} = \sum_{i=1}^{\ell_k}h_{[i]} - \sum_{i=j+1}^{\ell_k} h_{[i]} &\ge_{(a)} \hat{h}(\overline{D}_k^+) - \sum_{i=j+1}^{\ell_k} h_{[i]} \ge_{(b)} \hat{h}(\overline{D}_k^+) - \sum_{i=j+1}^{\ell_k} \underline{h}^\star \\
        &\ge_{(c)} \hat{h}(\overline{D}_k^+) - \sum_{i=j+1}^{\ell_k} \hat{h}_i =_{(d)} \sum_{i=1}^j \hat{h}_i.
    \end{align*}
    Inequality (a) is due to (\ref{eq:hh2}). Inequality (b) follows because $\underline{h}^\star \ge h_{[j]}$. 
    Inequality (c) follows because $\hat h_i\geq \underline{h}^\star$ for all $i\in [1,L]$. Equality (d) uses 
    the definition of $\hat{h}$. 
    
    This completes the proof of Proposition \ref{prop:intld} and Theorem \ref{theo:main}. 
\end{proof}

\begin{remark}\label{rem:eff} \rm We briefly comment on the computational complexity of finding the Lorenz optimal transfer. This transfer can be computed efficiently (i.e., polynomial time in the size of the input) using well-known combinatorial algorithms to compute a maximum-flow in a network. Recall that our two-stage procedure relies on (i) finding the Lorenz optimal fractional assignment in the initial network; and on (ii) finding a maximum-flow in an auxiliary network that is constructed using the knowledge of the Lorenz optimal fractional assignment. \cite{Meg77} describes an efficient algorithm for finding a lexicographically optimal flow in a network; this algorithm, in fact, computes a Lorenz optimal flow because a Lorenz optimal flow {\em is} lexicographically optimal.\footnote{The problem of finding a lex-optimal flow in a network is fundamental and has received a lot of attention in the network flow literature. Since ~\cite{Meg77}, several authors have identified faster algorithms for solving this problem, culminating in the work of~\cite{GalloGT89}, who describe how to pipeline various max-flow computations to find a lex-optimal flow in time that is asymptotically the same as finding a single max-flow.}

The auxiliary network can be constructed efficiently once a Lorenz optimal fractional assignment is known. Moreover, an integer Lorenz optimal assignment can be found by finding an integer max-flow in the auxiliary network that has upper and lower bounds on arc-flows, and this can be computed using any efficient combinatorial algorithm for finding a max-flow such as the Edmonds-Karp algorithm~\citep{EK72}.
\end{remark}

\section{Strategy-proofness}
\label{sec:SP}

In this section, we consider the problem of a planner who can observe the surpluses and deficits of schools but not the acceptable
schools of teachers. We show that the planner can nevertheless achieve the Lorenz dominant 
post-transfer deficit vector by incentivizing teachers to report their acceptable schools truthfully.

We assume that every teacher has ``trichotomous'' preferences: the top indifference class consists of her acceptable deficit schools, followed by the surplus school she is assigned to, and lastly, the set of non-acceptable deficit schools.
In particular, each teacher $t_i$ is indifferent between acceptable deficit schools in 
$A(t_i)$ and strictly prefers being transferred to a school in $A(t_i)$ 
than remaining in her initial assignment. Her initial assignment 
is strictly preferred to any school in $D\setminus A(t_i)$. As mentioned in the literature review, variants of such preferences have been used in \cite{M21} and \citet{ACEE21}. 

Let $A \equiv (A_1,\ldots,A_{|T|})$ denote a profile of acceptable sets (or schools). Let $\Sigma(A)$ denote the set of all transfers that result in the Lorenz dominant deficit vector. Although  $\Sigma(A)$ is non-empty (Theorem \ref{theo:main}),
it can contain multiple transfers. For instance, in the example in Figure \ref{fig:second}, the Lorenz dominant transfer $(1,1,2)$ can be achieved by three transfers. In all of them no teacher is transferred to school $d_1$. 
In the first, teachers $t_1$ and $t_2$ are transferred to $d_2$; in the second, teachers $t_1$ and $t_3$ are transferred to $d_2$; in the third, teacher $t_1$ is transferred to $d_2$ and teacher $t_3$ is transferred to $d_3$. 
Suppose that for every profile $A$, a transfer is selected from $\Sigma(A)$ according to a {\it fixed tie-breaking rule}.  
We show in Proposition \ref{prop:ic} that this mechanism is strategy-proof. 

A {\it matching} $\mu:T \rightarrow D$ is a mapping from set of teachers to deficit schools. A matching does not take into account acceptable schools of teachers or surplus constraints of surplus schools. A matching therefore differs from a transfer.
Let $\mathcal{M}$ be the set of all possible matchings and let $\rhd$ be an arbitrary strict ordering of the elements of  $\mathcal{M}$. Let $\max_\rhd(\Sigma(A))$ denote the maximal transfer in $\Sigma(A)$ according to $\rhd$.
A {\bf Lorenz dominant transfer mechanism (supplemented by tie-breaking $\rhd$)}, {\sc ldt}$_\rhd$, selects the transfer {\sc ldt}$_\rhd(A) := \max_\rhd(\Sigma(A))$ at every profile $A$. We denote the deficit school assigned to teacher $t_j \in T$ at profile $A$ by {\sc ldt}$_\rhd(t_j,A)$. Note that {\sc ldt}$_\rhd(t_j,A) \in A(t_j) \cup \{O(t_j)\}$ since our algorithm does not transfer any teacher to a non-acceptable school (according to reported preferences).

The Lorenz dominant mechanism {\sc ldt}$_\rhd(A)$ is {\bf strategy-proof} if for every teacher $t_j \in T$, every $A(t_j) \subseteq D$ and every $A(t_{-j})$ the following holds: 
\begin{align*}
\Big[{\rm LDT}_\rhd(t_j,A(t_j),A(t_{-j})) = O(t_j) \Big] \implies \Big[{\rm LDT}_\rhd(t_j,B,A(t_{-j})) \notin A(t_j)~\qquad \forall~B \subseteq D\Big].
\end{align*}
Strategy-proofness requires the following:  if a teacher is not transferred (remains in her initial surplus school) by reporting her true preference, then she {\sl cannot} be transferred to an acceptable school by reporting any other set of acceptable schools. It must be emphasized that a change in  a teacher's report of acceptable schools affects the set of feasible transfers, which in turn has a bearing on achievable deficit vectors.

\begin{prop}
    \label{prop:ic}
    The Lorenz dominant transfer mechanism {\sc ldt}$_\rhd$ is strategy-proof.
\end{prop}
\begin{proof}
        Pick $t_j \in T$ and a profile $A$ such that {\sc ldt}$_\rhd(t_j,A) = O(t_j)$. Let $B \subseteq D$ be any non-empty subset of deficit schools. Assume for contradiction that {\sc ldt}$_\rhd(t_j,(B,A(t_{-j}))) \in A(t_j)$. Since the acceptable sets of other teachers do not change from profile $A$ to profile $(B,A(t_{-j}))$ and {\sc ldt}$_\rhd(t_j,B,A(t_{-j})) \in A(t_j)$, it must be the case that {\sc ldt}$_\rhd(B,A(t_{-j}))$ is a feasible transfer in profile $A$. Similarly, since {\sc ldt}$_\rhd(t_j,A) = O(t_j)$, transfer {\sc ldt}$_\rhd(t_j,A)$ is a feasible transfer in profile $(B,A(t_{-j}))$. Hence, the deficit vectors generated by {\sc ldt}$_\rhd(t_j,(B,A(t_{-j})))$ and {\sc ldt}$_\rhd(t_j,A)$ Lorenz dominate each other, implying that both are Lorenz dominant transfers at both profiles. But, according to the
        definition  of {\sc ldt}$_\rhd$  we cannot pick {\sc ldt}$_\rhd(t_j,B,A(t_{-j}))$ at $(B,A(t_{-j}))$ when {\sc ldt}$_\rhd(t_j,A)$ is available. We have arrived at a contradiction.
\end{proof}

\begin{remark} \rm There is a more general way of making strategy-proof selections from the set $\Sigma(A)$. One could use a {\it choice function} that picks a transfer from every feasible subset of transfers. If the choice function satisfies the property of  {\sl contraction consistency}, (see \citet{Osbrub94}) an argument similar to that of the proof of Proposition \ref{prop:ic} can be used to show that the mechanism defined by $\Sigma(A)$ and the choice function, is strategy-proof.~\footnote{The use of choice function to make strategy-proof selections appears in \cite{NSSY22}. The greater generality of this approach over the fixed tie-breaking one, leads to a characterization result in their model (in conjunction with other axioms).}
\end{remark}

\section{Extensions}
\label{sec:ext}

In this section, we explore two extensions of our basic model.

\subsection{Allowing for transfers from surplus schools to surplus schools}

In our analysis so far, a teacher can only move from a 
surplus school (her initial assignment) to a deficit school. 
In this section, we relax this assumption. 
We consider an extension of our earlier model where a teacher can move from a surplus school (her initial assignment)
to either a deficit school or a surplus school. Each teacher $t_i$ has a set of acceptable 
schools $A(t_i)\subseteq D\cup S\setminus \{O(t_i)\}$. This set consists of all 
deficit and surplus schools that she would like to be transferred to. 
The other assumptions remain unchanged, i.e. no surplus school can 
become a deficit school post-transfer and no deficit school can become a surplus school 
post-transfer. Our goal is to show that our earlier analysis can be easily modified to 
cover the extended model.

We begin by describing the network flow formulation for the extended model.
This requires making some minor modifications to the original network described in 
Section \ref{sec:network}. In particular, the set of nodes 
and the edges in the original network remain as before. However,
some  ``backward  edges" are added to the original network.

Recall that the set of nodes $N$ in the graph consists of source $N_0$,
the set of surplus schools $S$, the set of teachers $T$ and 
the set of deficit schools $D$, where $D$ is the set of sinks. 
A directed capacity graph $G$ with edges $E$ is constructed as 
follows. There is an edge between the source node $N_0$ and 
each surplus school $s_j\in S$. There is an edge between each surplus school $s_j$ 
and each teacher $t_i$ such that $O(t_i)=s_j$. There is an edge from teacher $t_i$ to every acceptable school $A(t_i)$.
The capacity of an edge $(N_0,s_j)$ is $\alpha_j$. 
The capacity of an edge $(s_j,t_i)$, where $O(t_i) =s_j$, is $1$. 
Similarly the capacity of an edge $(t_i,d_k)$ where $d_k\in A(t_i)$, is $1$.
The capacity of an edge $(t_i,s_j)$ where $s_j\in A(t_i)$, is $1$. Each 
deficit school $d_k$ has a node capacity 
$\beta_k$. We let $G'$ denote this network. Notice the modified network is exactly like the earlier network $G$ except that the definition of acceptable schools has changed.

\begin{figure}[!hbt]
	\centering
	\includegraphics[width=3in]{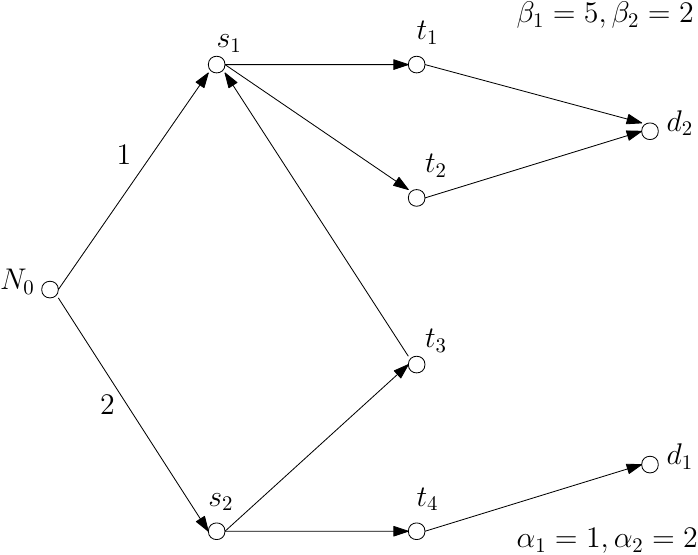}
	\caption{An example where teachers can have surplus schools acceptable}
	\label{fig:reform}
\end{figure}

We give an example to illustrate the  modified network. In this example, the sets of surplus schools, deficit schools and teachers are 
$S=\{s_1,s_2\}$, $D=\{d_1,d_2\}$ and  $T=\{t_1,t_2,t_3,t_4\}$ respectively. 
The initial assignment of teachers is given by $O(t_1)=O(t_2)=s_1$ 
and $O(t_3)=O(t_4)=s_2$. The acceptable sets are as follows: 
$A(t_1)=A(t_2)=\{d_2\}$, $A(t_3)=\{s_1\}$ and $A(t_4)=\{d_1\}$. 
The current deficit of schools $d_1$ and $d_2$ are $\beta_1=5$ and 
$\beta_2=2$. The surplus of schools $s_1$ and $s_2$ 
are $\alpha_1=1$ and $\alpha_2=2$. The modified network is shown in Figure \ref{fig:reform}.

In the original model where teachers can only be transferred to 
deficit schools, teacher $t_3$ is {\sl not} included in the network as 
she does not find any deficit school acceptable. The Lorenz dominant 
deficit vector is $(4,1)$ achieved by transferring one teacher, say 
$t_1$ from $s_1$ to $d_2$ and teacher $t_4$ from $s_2$ to $d_1$. 

In the extended model, teachers can be transferred either to deficit schools
or to surplus schools provided they are acceptable. Teacher $t_3$ is included in the network
with a forward edge $(s_2,t_3)$ and a backward edge $(t_3,s_1)$ respectively 
(see Figure \ref{fig:reform}). The Lorenz dominant deficit vector 
is $(3,1)$. This can be achieved by the following transfer: transfer $t_4$ to deficit 
school $d_1$, transfer $t_3$ to surplus school $s_1$ and 
transfer $t_1,t_2$ to deficit school $d_2$. It is now 
possible to transfer both teachers $t_1,t_2$ from 
surplus school $s_1$ since teacher $t_3$ is transferred 
to surplus school $s_1$. This ensures that surplus school 
$s_1$ does not become a deficit school post-transfer. 

A crucial observation is that the achievability result in \citet{megiddo1974optimal} continues to hold for the new network $G'$. As a result, we can define an integer- valued, convex cooperative game $w$ for the new network as well. The ``reduced form'' of the extended problem is identical to that of the original problem. 
Consequently, Theorem \ref{theo:main} and Proposition \ref{prop:ic} continue to hold for the extended model and the proofs remain unchanged. 

\subsection{Teacher specializations}

In the model so far, we have assumed that all teachers are homogeneous in their characteristics. 
In this extension, we consider a model where teachers have types depending on the subjects they can teach. 
If every teacher is qualified to teach only a single subject, a separate network can be constructed for each subject
and our earlier analysis applies to each subject. In this subsection, we consider
the case where a teacher may be qualified to teach several subjects but 
teaches only one subject in her assigned school. 
The goal of the planner is to transfer teachers respecting subject-wise surpluses and deficits, to equalize deficits across (school, subject) tuples. 
We show how our earlier analysis can be modified
to cover this case.

To simplify notation, we assume that there are only two subjects: Chemistry ($C$) and Physics ($P$). The set of teachers is $T$
and each  $t_i\in T$ has a type $\theta(t_i)$ that indicates which subjects she is qualified to teach. 
Here $\theta(t_i)\in \{C,P,CP\}$ where $C$ (or $P$) signifies that $t_i$ is qualified to teach 
only Chemistry (or only Physics) while $CP$ signifies that $t_i$ is qualified to teach both subjects.

There are three types of schools: {\it pure surplus schools} $S$ 
that have surplus teachers in both  subjects,  {\sl pure deficit schools} $D$ that have deficit of teachers in both subjects and {\sl mixed schools} $M$ that have a surplus in one subject and a deficit in the other.
For each $s_j \in S$,  $\alpha^C_j$ and $\alpha^P_j$ denote the surpluses of $s_j$ in $C$ and $P$ respectively. For each $d_k\in D$ , $\beta^C_k$ and $\beta^P_k$
denotes the deficits of $d_k$ in $C$ and $P$ by  respectively. For all $m_{\ell} \in M$,  $\alpha^x_{\ell}$ denotes the surplus of $m_{\ell}$ in subject $x$
and $\beta^y_{\ell}$ is the deficit of $m_{\ell}$ in subject $y$. We assume all surpluses and deficits are strictly positive integers.

Each $t_i \in T$ belongs to a unique school and teaches exactly one of the subjects she is qualified to teach in the school. Let $O(t_i)$ and $a(t_i)$ denote the school $t_i$ belongs to  and the subject she teaches there, respectively. 
We make some natural assumptions regarding the deployment of teacher $t_i$ of type $CP$. If $O(t_i) = d_k \in D$, then $t_i$ teaches the subject with the higher deficit in $d_k$; if $O(t_i) = m_{\ell} \in M$,
then $t_i$ teaches the subject for which $m_{\ell}$ has a deficit. This assumption ensures that deficits cannot be equalized by transfers {\it within} the same school.

In keeping with our earlier model, we assume that all teachers in $S$ can be transferred while no teacher in $D$ can be transferred. In school $m_{\ell} \in M$, the only
teachers who can be transferred are those of type $x$ where $\alpha^x_{\ell} >0$.\footnote{Recall that all teachers of type $CP$ in $m_{\ell}$ are teaching subject $y$ where $\beta^y_{\ell} >0$.}
We assume (i) no school with a surplus in a subject can have a deficit in that subject post transfer (ii)
no school with a deficit in a subject can have a surplus in that subject post transfer and (iii)
no surplus can grow larger, and no deficit can grow larger post transfer. 

Let $T'$ be the set of teachers who can be transferred. 
Each $t_i \in T'$  has a set of alternative schools that are acceptable to them:  $A(t_i)\subseteq D\cup M$. 
A network $G''$ is defined as follows.  

\begin{itemize}
\item The set of nodes in the graph consists of:

\begin{enumerate}
\item A source node $N_0$.
\item Two nodes for each $s_j\in S$: $(s_j,P)$ and $(s_j,C)$.
\item Two nodes for each $d_k\in D$: $(d_k,P)$ and $(d_k,C)$.
\item Two nodes for each  $m_{\ell}\in M$ that has a surplus in subject $x$ and 
deficit in subject $y$ (where $x,y$ are distinct and $x,y\in \{C,P\}$): $(m_{\ell},x)$ and $(m_{\ell},y)$.
\item A node for every teacher in $T'$.
\end{enumerate}

The nodes for pure deficit schools and the deficit ``subject'' nodes for mixed schools are
{\it sink} nodes. 

\item The set of edges $E$ and their capacities are as follows: 

\begin{enumerate}
\item Two directed edges $(N_0,(s_j,P))$ and $(N_0,(s_j,C))$ for every pure surplus school $s_j\in S$ with 
capacities $\alpha^C_j$ and $\alpha^P_j$ respectively.
\item A directed edge $(N_0,(m_{\ell},x))$ (where $x\in \{C,P\}$ is the subject in which school $m_{\ell}$ has a surplus) 
for every $m_{\ell}\in S$ with capacity $\alpha^x_{\ell}$.
\item For every $t_i\in T'$ with $O(t_i)=s_j\in S$, one of the following holds: a directed 
edge $((s_j,P),t_i))$ with capacity $1$ if $\theta(t_i)=P$, or a directed edge $((s_j,C),t_i))$ with 
capacity $1$ if $\theta(t_i)=C$, or a directed edge $((s_j,a(t_i)),t_i)$ with capacity $1$ if 
$\theta(t_i)=CP$ and $t_i$ teachers $a(t_i)$ in school $s_j$. 
\item For every $t_i\in T'$ with $O(t_i)=m_{\ell}\in M$ (where $m_{\ell}$ has a surplus in subject $x$),
we have a directed edge $((m_{\ell},x),t_i)$ with capacity $1$. 
\item For every $t_i\in T'$ with $O(t_i)\in S$, we have 
\begin{itemize}
\item for every pure deficit school $d_k\in A(t_i)\cap D$, one of the following holds:
a directed edge $(t_i,(d_k,C))$ with capacity $1$ if $\theta(t_i)=C$, or a directed edge $(t_i,(d_k,P))$ with capacity $1$ 
if $\theta(t_i)=P$, or two directed edges $(t_i,(d_k,C))$ and $(t_i,(d_k,P))$ with capacities $1$ if $\theta(t_i)=CP$.
\item for every mixed school $m_{\ell}\in A(t_i)\cap D$ with deficit in subject $y$ (sink node $(m_{\ell},y)$), a directed 
edge $(t_i,(m_{\ell},y))$ with capacity $1$ if either $\theta(t_i)=y$ or $\theta(t_i)=CP$.
\end{itemize}
\item For every $t_i\in T'$ with $O(t_i)=m_{\ell}\in M$ and mixed school $m_{\ell}$ has a surplus in subject $x$, we have a 
directed edge $(t_i,(\gamma,x))$ with capacity $1$ for every $\gamma\in A(t_i)$.\footnote{Since teacher $t_i$ 
belongs to a mixed school with surplus in subject $x$ and is included in the transfer process, we know 
by definition that type of teacher $t_i$ is $x$. Thus there is a directed edge between $t_i$ and sink 
node $(\gamma,x)$ for every school $\gamma$ in her acceptable set.}
\end{enumerate}

\item For every pure deficit school $d_k\in D$, the two associated sink nodes $(d_k,C)$ and $(d_k,P)$ 
have capacities $\beta_k^C$ and $\beta_k^P$ respectively.

\item For every mixed school $m_{\ell}\in M$ with a deficit in subject $y\in \{C,P\}$, the associated sink node 
$(m_{\ell},y)$ has a capacity $\beta_{\ell}^y$. 
\end{itemize}

Like our earlier analysis, 
we are interested in teacher reassignments that would result 
in a Lorenz dominant post-transfer deficit vector. 

It is routine to check that the achievability result in \citet{megiddo1974optimal} holds 
for the network $G''$. As before, we can define an integer-valued convex cooperative game for 
this new network as well. Theorem \ref{theo:main} and Proposition \ref{prop:ic} applied to this 
network establishes the existence of a Lorenz dominant integer transfer, which can be 
found using the MDR algorithm discussed earlier. 

\section{Other related literature} 
\label{sec:othrel}

Our paper is partly inspired by empirical work on teacher transfers in the state of Haryana, especially
 \citet{agarwal2018redistributing} and \citet{sharan2022teachers}. These papers look at data in Haryana and show the improvement of deficits in deficit schools as a result of transfers. \citet{agarwal2018redistributing} show that even local transfers (transfers up to a certain distance from the currently assigned school) leads to significant improvements in deficits. Our work is an attempt to address the issue of deficit reduction from a theoretical perspective. 

There is a recent theoretical literature on teacher transfers that is motivated by concerns different from ours. We comment briefly on some of these papers.

\citet{CTT22} study a model of teacher reassignment in the French school system. They consider a two-sided model where teachers and schools have preferences  and characterize reassignments that are ``maximally efficient and fair''. Unlike our model where teachers have to be transferred from surplus to deficit schools, a teacher in their model can be transferred from a school only if another teacher is transferred to that school. The algorithm they propose is a variant of the top-trading cycle that does not change the initial distribution of teachers. Indeed, improvements in the distribution of teachers is not a concern in the paper.

\citet{CDTTU25} study an extension of \citet{CTT22}  where each teacher has a type (from an ordered set of types) and each school is initially assigned a set of teachers. This creates a distribution of types at each school. They propose a family  of  indices to measure inequality of ``teacher quality'' across schools. They define a generalization of the top-trading cycle mechanism which is strategy-proof, ``constrained" Pareto efficient, and improves inequality from the status-quo in large markets. As in \citet{CTT22},
the initial distribution of teachers across schools remains unchanged post transfers although the distribution of teacher types in each schools can change.

\citet{HKYY25} study a similar model in the two-sided matching market and characterize different choice rules used by one side of the market that results in different distributional outcomes. In a related paper, \citet{HKY25} show that improving distributional objectives (for a family of such objectives) is incompatible with strategy-proofness, individual rationality and constrained efficiency. They introduce a new distributional objective called, ``pseudo $M^\natural$-concavity" (generalizing a well-known notion of discrete concavity called $M^\natural$-concavity), that is compatible with these axioms. They also propose a mechanism that satisfies these properties. \citet{EMY25} compare diversity in a school choice model using majorization technqiues. They axiomatically characterize choice rules that are consistent with a modified majorization notion which they introduce for school choice problems. In the standard allocation model with unit demand, \cite{ICDR} show that a particular variant of the top trading cycle Lorenz dominates the traditional top trading cycle, in terms of agent utilities. 

\cite{D19} consider a model where students have  already been admitted to a set of colleges. Students have to be transferred across colleges with constraints on ``imports" and ``exports" . In particular, they require that these exchanges are {\sl balanced}, implying that exports equal imports at each college. Students have strict preferences over colleges and colleges have priorities over students. They introduce the ``two-sided top-trading-cycles" mechanism and show that it is the unique mechanism that is priority respecting, balanced-efficient, student-strategy-proof, acceptable and individually rational. It is clear that our model is different from theirs. Importantly,  there are also no balancing constraints in our model.

There is a recent literature on matching with constraints: see \cite{K15,K24} and references therein. This literature considers the two-sided matching (school choice) problem, where teachers have preferences (strict orders) over schools and schools have priorities over students. In addition to these, there are distributional constraints, which are modelled as restrictions on subsets of students that a school can admit. The objectives in these papers vary: it is {\sl fairness} (no justified envy) in \cite{K24}; while it is efficiency in \cite{K15}. In contrast, our objective is to achieve a Lorenz dominant post-transfer deficit vector assuming teachers preferences are restricted. We note that our restrictions on transfers from surplus to deficit schools cannot be modelled as constraints on matching as in these papers.

\cite{A24} considers the school choice problem with the option of redistributing additional seats. In their model, schools can redistribute seats between themselves to improve efficiency based on the reported preferences of the students and given the priorities of schools. They introduce a simple class of algorithms that characterizes the set of efficient matchings in their model.

We model teacher preferences as being either dichotomous or trichotomous.  Preferences of this nature appear frequently in the matching literature. We have already noted their use in 
\citet{BM04}.  \cite{M21} analyze an ``object reallocation" model where agents are endowed with bundles of objects and their preferences over bundles of objects satisfy a trichotomous structure. They define a class of individually rational, Pareto-efficient, and strategy-proof mechanisms for this model.  \citet{ACEE21} study a model of exchange of goods, where each agent can receive goods from other agents. The preference of an agent over the bundles of goods received are dichotomous in a particular way: each agent partitions the set of agents as acceptable and unacceptable, and bundles of goods from acceptable agents (satisfying some upper bound on quantity) are preferred over other bundles. They define a class of mechanisms called priority mechanisms and show them to be strategy-proof. \citet{ASY19} study a model where agents need to be matched in pairs to complete a project. Each agent has preferences over partner, project pairs that are separable. Marginal preferences over partners and projects are dichotomous. In other words, each agent has an acceptable set of partners and an acceptable set of projects. They propose a mechanism in this setting and show that it is weakly stable and strategy-proof. We note that none of these papers are concerned with the issues that we are.


\appendix 

\section{Proof of Lorenz domination of the MDR algorithm}
\label{app:mdr}

\begin{theorem}[\citet{dutta1989concept}, \citet{megiddo1974optimal}]
    \label{theo:mdr}
    The MDR algorithm produces a Lorenz dominant transfer.
\end{theorem}

\begin{remark} \rm Note that \citet{dutta1989concept} show that the MDR algorithm produces a deficit vector that Lorenz dominates every achievable deficit vector corresponding to maximum flow.\footnote{In their terminology, the provide a Lorenz dominant core point of a convex game. Since core constraints are binding for the grand coalition and the achievability constraints are not necessarily binding, we cannot directly apply their result.} It is not clear that if a such deficit vector will Lorenz dominate achievable deficit vectors that {\sl do not} correspond to a maximum flow. The current proof provides this minor extension. \citet{megiddo1974optimal} did not consider Lorenz domination but {\sl lex-optimality}.
\end{remark}

The proof of Theorem \ref{theo:mdr} uses the following claim.    
    \begin{claim}[\citet{HLP29}]
    \label{cl:hlp}
        Let $a \equiv (a_1,\ldots,a_{\ell})$ be a vector and $a^{{\rm avg}} := \frac{a_1+\ldots+a_{\ell}}{\ell}$. Then, $(a^{{\rm avg}},\ldots,a^{{\rm avg}})$ Lorenz dominates $a$.
    \end{claim}
    We give a proof for completeness here though it is a straightforward consequence of Fact \ref{fact:hlp}, which is due to \citet{HLP29}.
    
    \begin{proof}
        Without loss of generality, assume $a_1 \ge \ldots \ge a_{\ell}$. Pick $k \in \{1,\ldots,\ell\}$. Suppose $k$ is such that $a_{k-1} \ge a^{{\rm avg}} \ge a_k$. Clearly, for any $j \le k-1$, 
        $\sum_{i=1}^j a_i \ge j a^{{\rm avg}}.$ 
        Now, for any $j \ge k$, we see that 
        $$\sum_{i=1}^j a_i = \ell a^{{\rm avg}} - \sum_{i=j+1}^{\ell} a_i \ge \ell a^{{\rm avg}} - (\ell-j)a^{{\rm avg}} = j a^{{\rm avg}},$$ 
        where the inequality follows because $a^{{\rm avg}} \ge a_i$ for all $i \ge k$. This completes the proof.
    \end{proof}

\noindent {\sc Proof of Theorem \ref{theo:mdr}:}

\begin{proof}
We prove a series of claims. First, we show how $h^\star$ is ordered. Clearly, if $d_i,d_j \in D_k$ for some $k$, then $h^\star_i=h^\star_j$.
\begin{claim}
    \label{cl:orderld}
    Suppose $d_i \in D_k$ and $d_j \in D_{k+1}$. Then $h^\star_i \ge h^\star_j$.
\end{claim}
\begin{proof}
    Consider Step $k$ of the MDR algorithm where $D_k$ is chosen as the argmax and $D_k\cup D_{k+1}$ is a feasible option. Then by definition,
    \begin{align*}
        \frac{w(\overline{D}_k) - w(\overline{D}_{k-1})}{|D_k|} &\ge \frac{w(\overline{D}_{k+1}) - w(\overline{D}_{k-1})}{|D_k|+|D_{k+1}|} \\
        \implies |D_{k+1}| \Big( w(\overline{D}_k) - w(\overline{D}_{k-1}) \Big) &\ge |D_k| \Big(w(\overline{D}_{k+1}) - w(\overline{D}_{k}) \Big) \\
        \implies h^\star_i = \frac{w(\overline{D}_k) - w(\overline{D}_{k-1})}{|D_k|} &\ge \frac{w(\overline{D}_{k+1}) - w(\overline{D}_{k})}{|D_{k+1}|} = h^\star_j
    \end{align*}
\end{proof}

The next claim says $h^\star$ is achievable.
\begin{claim}
    \label{cl:maxach}
    The deficit vector $h^\star$ is achievable.
\end{claim}
\begin{proof}
    By construction, $\sum_{d_i \in D}h^\star_i = w(D)$. Hence, pick any $B \subsetneq D$. Partition $B$ into $(B_1,\ldots,B_{k^\star})$ such that $B_1 \subseteq D_1, B_2 \subseteq D_2, \ldots, B_{k^\star} \subseteq D_{k^\star}$. Denote $|B_j| \equiv \ell_j$ for each $j \in \{1,\ldots,k^\star\}$ and note that some $\ell_j$ can be equal to zero also. Then,
    \begin{align*}
        \sum_{d_i \in B}h^\star_i = \sum_{j=1}^{k^\star} \sum_{d_i \in B_j} h^\star_i &= \sum_{j=1: \ell_j \ne 0}^{k^\star} \ell_j \frac{w(\overline{D}_j) - w(\overline{D}_{j-1})}{|D_j|} \\
        &\ge_{(a)} \sum_{j=1}^{k^\star} \ell_j \frac{w(\overline{D}_{j-1} \cup B_j) - w(\overline{D}_{j-1})}{\ell_j} \\
        &\ge_{(b)} \sum_{j=1}^{k^\star} \Big[ w(\overline{B}_j) - w(\overline{B}_{j-1})\Big] ~\qquad~\textrm{($\overline{B}_j=\cup_{i=1}^jB_i$ and $\overline{B}_{j-1}=\cup_{i=1}^{j-1}B_i$)}\\
        &= w(B)
    \end{align*}
    where the first inequality (a) is by construction of $D_j$:
    \begin{align*}
        \frac{w(\overline{D}_j) - w(\overline{D}_{j-1})}{|D_j|} \ge \frac{w(\overline{D}_{j-1} \cup B) - w(\overline{D}_{j-1})}{|B|} ~\qquad~\forall~B \subseteq D \setminus \overline{D}_{j-1}
    \end{align*}
    and the second inequality (b) is due to convexity.
    This shows that $h^\star$ is achievable.
\end{proof}

\begin{claim}
    \label{cl:contld}
 The deficit vector $h^\star$ Lorenz dominates every achievable deficit vector.
\end{claim}
\begin{proof}
    Let $h$ be an arbitrary achievable deficit vector. Pick $k \in \{1,\ldots,L\}$. Suppose the $k$-th highest deficit in $h^\star$ lies in $D_j$. Let $k'$ be the number of components in $D_j$ and rest $(k-k')$ components be in $\overline{D}_{j-1}$. By definition
    \begin{align}
        \sum_{i=1}^k h^\star_{[i]} &= w(\overline{D}_{j-1}) + k'\frac{w(\overline{D}_j) - w(\overline{D}_{j-1})}{|D_j|} \nonumber \\
        &= \frac{1}{|D_j|}\Big[ (|D_j|- k')w(\overline{D}_{j-1}) + k' w(\overline{D}_j)\Big] \nonumber \\
        &\le \frac{1}{|D_j|}\Big[ (|D_j|- k')h(\overline{D}_{j-1}) + k' h(\overline{D}_j)\Big] \nonumber \\
        &= h(\overline{D}_{j-1}) + \frac{k'}{|D_j|} \Big( h(\overline{D}_j) - h(\overline{D}_{j-1}) \Big) \nonumber \\
        &= h(\overline{D}_{j-1}) + k'\frac{h(D_j)}{|D_j|}, \label{eq:con1}
    \end{align}
    where the first equality is by definition of $h^\star$, the inequality is by achievability of $h$ and the fact that $k' \le |D_j|$, and the final equality in (\ref{eq:con1}) is due to the fact that $h$ is linear (implying $h(D_j) = h(\overline{D}_j) - h(\overline{D}_{j-1})$).

    Suppose there are $\ell$ deficit schools in $D_j$. Denote the ordered deficits of $h$ in $D_j$ as $h^{D_j}_{[1]} \ge \ldots \ge h^{D_j}_{[\ell]}$. By Claim \ref{cl:hlp}, we see that for any $k' \in \{1,\ldots,\ell\}$,
    \begin{align*}
        k'\frac{h(D_j)}{|D_j|} \le \sum_{i=1}^{k'} h^{D_j}_{[i]}
    \end{align*}
    Using this with inequality (\ref{eq:con1}), we get 
    \begin{align*}
        \sum_{i=1}^k h^\star_{[i]} \le h(\overline{D}_{j-1}) + \sum_{i=1}^{k'} h^{D_j}_{[i]} \le \sum_{i=1}^k h_{[i]}.
    \end{align*}
    This establishes the desired claim.
\end{proof}

These claims establish the achievability and Lorenz dominance of $h^\star$.
\end{proof}

\end{document}